\documentclass[11pt,a4paper,reqno]{amsart}%
\usepackage{amsthm,amsmath,amsfonts,amssymb,amsxtra,appendix,bookmark,dsfont,bm}

\theoremstyle{plain}
\newtheorem{theorem}{Theorem}
\newtheorem{lemma}[theorem]{Lemma}

\theoremstyle{definition}

\newtheorem*{example}{Example}

\theoremstyle{remark}

\DeclareMathOperator{\Tr}{Tr}
\DeclareMathOperator{\Ker}{Ker}

\DeclareMathOperator{\Span}{Span}

\def\bq{\begin{eqnarray}}
\def\eq{\end{eqnarray}}
\def\bqq{\begin{align*}}
\def\eqq{\end{align*}}

\def\nn{\nonumber}

\def\eps{\varepsilon}
\def\wto{\rightharpoonup}

\renewcommand{\epsilon}{\varepsilon}
\newcommand\1{{\ensuremath {\mathds 1} }}

\def\bU{\mathbb{U}}

\def\cF {\mathcal{F}}

\def\cV {\mathcal{V}}
\def\R {\mathbb{R}}
\def\C {\mathbb{C}}

\def\U {\mathbb{U}}
\def\cU {\mathcal{U}}
\def\cJ {\mathcal{J}}

\def\cA{\mathcal{A}}

\def\F {\mathcal{F}}

\def\gH{\mathfrak{h}}

\def\h{\gH}

\def\cV {\mathcal{V}}
\def\R {\mathbb{R}}
\def\C {\mathbb{C}}

\def\U {\mathbb{U}}
\def\J {\mathcal{J}}
\def\cJ {\mathcal{J}}
\def\S {\mathcal{S}}

\def\cA{\mathcal{A}}

\newcommand{\bH}{\mathbb{H}}
\newcommand{\dGamma}{{\ensuremath{\rm d}\Gamma}}
\title[Diagonalization of quadratic Hamiltonians]{Diagonalization of bosonic quadratic Hamiltonians by Bogoliubov transformations}

\author[P.T. Nam]{Phan Th\`anh Nam}
\address{Institute of Science and Technology Austria, Am Campus 1, 3400 Klosterneuburg, Austria} 
\email{pnam@ist.ac.at}

\author[M. Napi\'orkowski]{Marcin Napi\'orkowski}
\address{Institute of Science and Technology Austria, Am Campus 1, 3400 Klosterneuburg, Austria} 
\email{marcin.napiorkowski@ist.ac.at}

\author[J.~P. Solovej]{Jan Philip Solovej}
\address{Department of Mathematical Sciences, University of Copenhagen, Universitetsparken 5, DK-2100 Copenhagen \O, Denmark} 
\email{solovej@math.ku.dk}

\begin{document}

\begin{abstract} We provide general conditions for which bosonic quadratic Hamiltonians on Fock spaces can be diagonalized by Bogoliubov transformations. Our results cover the case when quantum systems have infinite degrees of freedom and the associated one-body kinetic and paring operators are unbounded. Our sufficient conditions are  optimal in the sense that they become necessary when the relevant one-body operators commute.    
\end{abstract}

\date{\today}

\maketitle

\section{Introduction}

We consider Hamiltonians on Fock space which are quadratic in terms of bosonic creation and annihilation operators. In many cases the quadratic Hamiltonians can be diagonalized by Bogoliubov transformations, namely they can be transformed to those of noninteracting particles by a special class of unitary operators which preserve the CCR algebra. The aim of our present work is to 
give rigorous conditions for which the diagonalization can be carried out for quantum systems of infinite degrees of freedom where the kinetic and paring operators are unbounded. 

\subsection{Quadratic Hamiltonian}  
Let us introduce the mathematical setting. Our one-body Hilbert space $\h$ is a complex separable Hilbert space with inner product $\langle.,.\rangle$ which is linear in the second variable and anti-linear in the first. In the grand canonical ensemble the number of particles is not fixed and it is natural to introduce the (bosonic) Fock space  
$$\F(\h):=\bigoplus_{N=0}^\infty \bigotimes_{\text{sym}}^N \h = \C \oplus \gH \oplus \left(\gH\otimes_s \gH \right) \oplus \cdots$$

Noninteracting systems are described by the Hamiltonians of the form 
$$
\dGamma(h):= \bigoplus_{N=0}^\infty (\sum_{j=1}^N h_j) =  0 \oplus h \oplus (h\otimes 1 + 1 \otimes h) \oplus \cdots  
$$
on the Fock space, where $h>0$ is a self-adjoint operator on $\gH$. Although we will work in an abstract setting, the reader may keep in mind the typical example that $h=-\Delta+V(x)$ on $\gH=L^2(\R^d)$, where $V$ is an external potential which serves to bind the particles. 
The operator $\dGamma(h)$ is well-defined on the core 
$$
\bigcup_{M\ge 0} \bigoplus_{n=0}^M  \bigotimes_{\rm sym}^n D(h)
$$
and it can be extended to a positive self-adjoint operator on Fock space by Friedrichs' extension. The spectrum of $\dGamma(h)$ is nothing but the closure of the finite sums of elements of the spectrum of $h$. In particular, the spectrum of the particle number operator $\mathcal{N}:=\dGamma(1)$ is $\{0,1,2,...\}$.

In many physical situations the interaction between particles plays a crucial role and it
complicates the picture dramatically. In principle, solving interacting systems exactly is mostly unrealistic and certain approximations are necessary. In the celebrated 1947 paper \cite{Bogoliubov-47b}, Bogoliubov introduced an approximation theory for a weakly interacting Bose gas where the many-body system is effectively described by a {\em quadratic Hamiltonian} on Fock space, which will be described below. We refer to the book \cite{LieSeiSolYng-05} for a pedagogical introduction to Bogoliubov's approximation. Bogoliubov's theory has been justified rigorously in various situations including the ground state energy of one and two-component Bose gases~\cite{LieSol-01,LieSol-04,Solovej-06}, the Lee-Huang-Yang formula of dilute gases~\cite{ErdSchYau-08,GiuSei-09,YauYin-09} and the excitation spectrum in the mean-field limit  \cite{Seiringer-11,GreSei-13,LewNamSerSol-15,DerNap-13,NamSei-14}. 

To describe quadratic Hamiltonians on bosonic Fock space, we introduce the {\em creation} and {\em annihilation} operators. For any vector $f \in \gH$, the creation operator $a^*(f)$ and the annihilation operator $a(f)$ are defined by the following actions  
\begin{align*} 
{a^*}({f_{N+1}})\left( {\sum\limits_{\sigma  \in {S_{N}}} {{f_{\sigma (1)}} } \otimes ... \otimes {f_{\sigma (N)}}} \right) = \frac{1}{\sqrt {N+1}} \sum\limits_{\sigma  \in {S_{N+1}}} {{f_{\sigma (1)}} }  \otimes ... \otimes {f_{\sigma (N+1)}},
\end{align*}
\begin{align*} 
  a(f_{N+1})\left( {\sum\limits_{\sigma  \in {S_N}} {{f_{\sigma (1)}} } \otimes ... \otimes {f_{\sigma (N)}}} \right) = \sqrt N \sum\limits_{\sigma  \in {S_N}} {\left\langle {f_{N+1},{f_{\sigma (1)}}} \right\rangle } {f_{\sigma (2)}} \otimes ... \otimes {f_{\sigma (N)}},
\end{align*}
for all $f_1,...,f_{N+1}$ in $\gH$, and all $N=0,1,2,...$. These operators satisfy the canonical commutation relations (CCR)
\bq \label{eq:CCR}
[a(f),a(g)]=0, \quad [a^*(f),a^*(g)]=0, \quad [a(f),a^*(g)]=\langle f, g\rangle, \quad \forall f,g\in \gH.
\eq
In particular, for every $f \in \gH$ we have 
\bq
a(f)|0\rangle =0 \nn 
\eq
where $|0\rangle=1\oplus 0 \oplus 0 \cdots$ is the Fock space vacuum.

In general, a quadratic Hamiltonian on Fock space is a linear operator which is quadratic in terms of creation and annihilation operators. For example, $\dGamma(h)$ is a quadratic Hamiltonian because we can write
\begin{align*}
\dGamma(h)= \sum_{m,n\ge 1}\langle f_m, h f_n \rangle a^*(f_m) a(f_n) 
\end{align*}
where $\{f_n\}_{n\ge 1}\subset D(h)$ is an arbitrary orthonormal basis for $\gH$ (the sum on the right side is independent of the choice of the basis). In this paper, we will consider a general quadratic operator of the form
\bq \label{eq:quadratic-Hamiltonian}
\bH = \dGamma(h) + \frac{1}{2}\sum_{m,n\ge 1} \left( \langle J^*kf_m,f_n \rangle a(f_m) a(f_n) + \overline{\langle J^*kf_m,f_n \rangle} a^*(f_m)a^*(f_n)  \right)
\eq
where $k:\gH\to \gH^*$ is an unbounded linear operator with $D(h)\subset D(k)$ (called {\em pairing operator}) 
and $J:\gH\to \gH^*$ is the anti-unitary operator \footnote{If $C:\gH\to \mathfrak{K}$ is anti-linear, then $C^*: \mathfrak{K} \to \gH$ is defined by $\langle C^*g,f \rangle_{\gH} = \langle C f, g \rangle_{\mathfrak{K}}$ for all $f\in \gH,g\in \mathfrak{K}.$ The anti-linear map $C$ is an anti-unitary if $C^*C=1_{\gH}$ and $CC^*=1_{\mathfrak{K}}$.} defined by
$$J(f)(g)=\langle f,g \rangle,\quad \forall f,g\in \gH.$$  
Since $\bH$ remains the same when $k$ is replaced by $(k+J k^* J)/2$, we will always assume without loss of generality that 
\bq \label{eq:symmetry-b-2}
k^* = J^* k J^*.
\eq

In fact, the formula \eqref{eq:quadratic-Hamiltonian} is formal but $\bH$ can be defined properly as a quadratic form as follows. For every normalized vector $\Psi \in \cF(\gH)$ with finite particle number expectation, namely $\langle \Psi, \mathcal{N} \Psi \rangle <\infty$, its one-particle density matrices $\gamma_\Psi: \h\to \h$ and $\alpha_\Psi:\h \to \h^*$ are linear operators defined by
\bq \label{eq:def-gamma-alpha}
\left\langle {f,{\gamma _\Psi }g} \right\rangle  = \left\langle \Psi, {{a^*}(g)a(f)} \Psi \right\rangle,~~ \left\langle {Jf, \alpha _\Psi  g } \right\rangle  = \left\langle \Psi, {a^*(g)a^*(f)} \Psi\right\rangle, \quad\forall f,g\in \h.
\eq
A formal calculation using \eqref{eq:quadratic-Hamiltonian} leads to the expression 
\bq \label{eq:def-H-gamma-alpha}
\langle \Psi, \bH \Psi \rangle = \Tr(h^{1/2} \gamma_\Psi h^{1/2})+ \Re \Tr(k^* \alpha_\Psi).
\eq
The formula \eqref{eq:def-H-gamma-alpha} makes sense when $h^{1/2} \gamma_\Psi h^{1/2}$ and $k^* \alpha_\Psi$ are trace class operators. We will use \eqref{eq:def-H-gamma-alpha} to define $\bH$ as a quadratic form with a dense form domain described below.

Since $\gH$ is separable and $D(h)$ is dense in $\gH$, we can choose finite-dimensional subspaces $\{Q_n\}_{n=1}^\infty$ such that
$$ Q_1 \subset Q_2 \subset ... \subset D(h) \quad \text{and}\quad \overline{\bigcup_{n\ge 1} Q_n} = \gH.$$
Then it is straightforward to verify that  
\begin{align} \label{eq:def-mcQ} \mathcal{Q}:=\bigcup_{M\ge 0} \bigcup_{n\ge 1} \Big( \bigotimes_{\rm sym}^M Q_n \Big) 
\end{align}
is a dense subspace of $\cF(\gH)$. Moreover, for every normalized vector $\Psi \in \mathcal{Q}$, $\gamma_\Psi$ and $\alpha_\Psi$ are finite rank operators with ranges in $D(h)$ and $J D(h)$, respectively. Note that $JD(h)\subset D(k^*)$ because $D(h)\subset D(k)$ and $k^*=J^*kJ^*$. Thus $h^{1/2} \gamma_\Psi h^{1/2}$ and $k^* \alpha_\Psi$ are trace class and $\langle \Psi, \bH \Psi \rangle$ is well-defined by \eqref{eq:def-H-gamma-alpha} for every normalized vector $\Psi$ in $\mathcal{Q}$. 

Under certain conditions (see Lemma \ref{lem:bounded-below-H}), we can show that the quadratic form $\bH$ defined by \eqref{eq:def-H-gamma-alpha} is bounded from below and closable, and hence its closure defines a self-adjoint operator by \cite[Theorem VIII.15]{ReeSim-80}. Let us mention that if $k$ is not Hilbert-Schmidt, then the vacuum does not belong to the operator domain of $\bH$ (although it belongs to the form domain). Therefore, it is not easy to define $\bH$ as an operator with a dense domain at the beginning. 
 
A key feature of the quadratic Hamiltonians in Bogoliubov's theory is that they can be diagonalized to those of {noninteracting} systems by a special class of unitary operators which preserve the CCR algebra. By Bogoliubov's argument \cite{Bogoliubov-47b}, the diagonalization problem on Fock space can be associated to a diagonalization problem on $\gH\oplus \gH^*$ in a very natural way which will be  briefly recalled below.  


\subsection{Bogoliubov transformation} \label{sec:Bogoliubov-transformation} Since we will consider transformations on $\gH\oplus \gH^*$, it is convenient to introduce the {generalized annihilation} and {creation} operators 
\bq \label{eq:def-A=a+a*}
A(f\oplus Jg)=a(f)+a^*(g), \quad A^*(f\oplus Jg)=a^*(f)+a(g),\quad \forall f,g\in \h.
\eq
They satisfy the conjugate and canonical commutation relations 
\bq \label{eq:conjugate-CCR}
A^*( F_1)=A(\J F_1),\quad \left[ {A(F_1 ),A^* (F_2 )} \right] = (F_1 ,\S F_2 ),\quad \forall F_1,F_2\in  \h\oplus \h^*
\eq 
where we have introduced the block operators on $\gH\oplus \gH^*$
\bq
\S = \left( \begin{gathered}
  1~{\text ~~~}0 \hfill \\
  0~~-1 \hfill \\ 
\end{gathered}  \right),\quad \J = \left( {\begin{array}{*{20}c}
   0 & {J^* }  \\
   J & 0  \\
 \end{array} } \right). \label{eq:def-S-J}
\eq
Note that $S=S^{-1}=S^*$ is a unitary on $\gH\oplus \gH^*$ and $\cJ=\cJ^{-1}=\cJ^{*}$ is an anti-unitary. The symplectic matrix $S$ is the bosonic analogue to the identity in the fermionic case. 

We say that a bounded operator $\cV$ on $\gH\oplus \gH^*$ is {\em unitarily implemented} by a unitary operator $\bU_\cV$ on Fock space if 
\bq \label{eq:V-Uv-action}
\bU_\cV A(F) \bU_\cV^* = A(\cV F), \quad \forall F\in \gH\oplus \gH^*.
\eq  
It is easy to see that if \eqref{eq:V-Uv-action} holds true, then the CCR \eqref{eq:conjugate-CCR} imply the following compatibility conditions 
\begin{align}\label{eq:Bogoliubov-transformation}\cJ \cV\cJ=\cV, \quad \cV^* S \cV = S = \cV S \cV^*.
\end{align}
Any {bounded} operator $\cV$ on $\gH\oplus \gH^*$ satisfying \eqref{eq:Bogoliubov-transformation} is called a {\em Bogoliubov transformation}. See \cite[Chap. 11]{DerGer-13} for an alternative description of Bogoliubov transformations in the context of symplectic geometry.

The condition $\cJ \cV \cJ =\cV$ means that $\cV$ has the block form
\bq \label{eq:u-v-block}
\cV = \left( {\begin{array}{c c}
  U & J^* V J^*  \\ 
  V & J U J^* 
\end{array}} \right)
\eq
where $U:\gH \to \gH$ and $V:\gH \to \gH^*$ are linear bounded operators. Under this form, the condition $\cV^* S \cV = S = \cV S \cV^*$ is equivalent to 
\bq \label{eq:relation-U-V}
U^* U = 1+ V^*V, \quad UU^*= 1+ J^* VV^* J, \quad  V^* J U =  U^*J^*V.
\eq
It is a fundamental result that a Bogoliubov transformation $\cV$ of the form \eqref{eq:u-v-block}  is unitarily implementable if and only if it satisfies {\em Shale's condition} \cite{Shale-62} 
\bq \label{eq:Shale}
\|V\|_{\rm HS}^2 = \Tr(V^*V)<\infty.
\eq 

Now we come back to the problem of diagonalizing $\bH$. Using the formal formula \eqref{eq:quadratic-Hamiltonian} and the assumption $k^*=J^*kJ^*$, we can write 
\bq \label{eq:quadratic-Hamiltonian-def2}
\bH= \bH_{\cA} - \frac{1}{2}\Tr(h) 
\eq
where
\begin{align} \label{eq:def-A-block}
\cA := \left( {\begin{array}{*{20}{c}}
   h & k^* \\ 
  k & J h J^* 
\end{array}} \right) 
\end{align}
and 
$$
\bH_{\cA}:=\frac{1}{2}\sum_{m,n\ge 1} \langle F_m, \cA F_n \rangle A^*(F_m)A(F_n).
$$
Here $\{F_n\}_{n\ge 1}$ is an orthonormal basis for $\gH\oplus \gH^*$ and the definition $\bH_{\cA}$ is  independent of the choice of the basis. Note that $\cJ\cA\cJ=\cA$ because of the symmetry condition $k^*=J^*kJ^*$. 

Now let $\cV$ be a Bogoliubov transformation on $\gH\oplus \gH^*$ which is implemented by a unitary operator $\bU_\cV$ on Fock space as in \eqref{eq:V-Uv-action}. Then we can verify that $\bU_\cV \bH_{\cA} \bU_\cV^* =  \bH_{\cV \cA \cV^*}$ and hence \eqref{eq:quadratic-Hamiltonian-def2} is equivalent to 
\bq \label{eq:quadratic-Hamiltonian-def3}
\bU_{\cV}\bH \bU_{\cV}^* = \bH_{\cV\cA \cV^*} - \frac{1}{2}\Tr(h). 
\eq
In particular, if $\cV \cA \cV^*$ is {\em block diagonal}, namely
$$ 
\cV \cA \cV^* = \left( {\begin{array}{*{20}{c}}
  \xi & 0  \\ 
  0& J \xi J^* 
\end{array}} \right)
$$
for some operator $\xi:\gH\to \gH$, then \eqref{eq:quadratic-Hamiltonian-def3} reduces to 
\bq \label{eq:quadratic-Hamiltonian-def4}
\bU_{\cV}\bH \bU_{\cV}^* = \dGamma(\xi)+ \frac{1}{2} \Tr( \xi -h). 
\eq
Note that all formulas \eqref{eq:quadratic-Hamiltonian-def2}, \eqref{eq:quadratic-Hamiltonian-def3} and \eqref{eq:quadratic-Hamiltonian-def4} are formal because $h$, $\xi$ and $\xi-h$ may be not trace class. Nevertheless, the above heuristic argument suggests that the diagonalization problem on $\bH$ can be reduced to the diagonalization problem on $\cA$ by Bogoliubov transformations. 

\subsection{Diagonalization conditions} \label{sec:example} In this paper we are interested in the conditions on $h$ and $k$ such that $\cA$ and $\bH$ can be diagonalized rigorously. 

Let us make some historical remarks. The physical model in Bogoliubov's 1947 paper \cite{Bogoliubov-47b} corresponds to the case when $\dim \gH=2$ and $\cA$ is a $2\times 2$ real matrix which can be diagonalized explicitly (more precisely, in his case particles only come in pairs with momenta $\pm p$ and each pair can be diagonalized independently). In fact, when $\dim \gH$ is finite, the  diagonalization of $\cA$ by symplectic matrices can be done by Williamson's Theorem \cite{Williamson-36}. We refer to H\"ormander \cite{Hormander-95} for a complete discussion on the diagonalization problem in the finite dimensional case.

In the 1950's and 1960's, Friedrichs \cite{Friedrichs-53} and Berezin \cite{Berezin-66} gave general diagonalization results in the case $\dim \gH=+\infty$, assuming that $h$ is bounded, $k$ is Hilbert-Schmidt, and $\cA\ge \mu>0$ for a constant $\mu$. Note that the gap condition $\cA\ge \mu$ requires that $h\ge \mu>0$.

In the present paper we always assume that $\cA>0$ but we do not require a gap. In some cases, the weaker assumption $\cA\ge 0$ might be also considered, but is is usually transferred back to the strict case $\cA>0$ by using an appropriate decomposition; see Kato and Mugibayashi \cite{KatMug-67} for a further discussion.

In many physical applications, it is important to consider unbounded operators. In the recent works on the excitation spectrum of interacting Bose gases, the diagonalization problem has been studied by Grech and Seiringer in \cite{GreSei-13} when $h$ is a positive operator with compact resolvent and $k$ is Hilbert-Schmidt, and then by Lewin, Nam, Serfaty and Solovej \cite[Appendix A]{LewNamSerSol-15} when $h$ is a general unbounded operator satisfying $h\ge \mu>0$. 


Very recently, in 2014, Bach and Bru \cite{BacBru-15} established for the first time the diagonalization problem when $h$ is not bounded below away from zero. They assumed that $h>0$, $\|kh^{-1}\|<1$ and $kh^{-s}$ is Hilbert-Schmidt for all $s\in [0,1+\eps]$ for some $\eps>0$ (see conditions (A2) and (A5) in \cite{BacBru-15}). 

In the present paper, we relax not only the gap condition $h\ge \mu>0$ but also the Hilbert-Schmidt conditions on $k$ ($k$ is even allowed to be an unbounded operator). Our conditions are motivated by the following simple example where all relevant operators commute. 
\begin{example}[Commutative case] Let $h$ and $k$ be multiplication operators on $\gH=L^2(\Omega,\C)$, for some measure space $\Omega$. Then $J$ is simply complex conjugation and we can identify $\gH^*=\gH$ for simplicity. Assume that $h>0$, but $k$ is not necessarily real-valued. Then 
$$
\cA := \left( {\begin{array}{*{20}{c}}
   h & k \\ 
  k & h 
\end{array}} \right)>0 \quad \text{on}~\gH\oplus \gH^*.
$$
if and only if $-1<G<1$ with $G:=|k|h^{-1}$. In this case, $\cA$ is diagonalized by the linear operator
$$
\cV :=  \sqrt{\frac{1}{2}+\frac{1}{2\sqrt{1-G^2}}} \left( {\begin{array}{*{20}{c}}
  1 & \frac{-G}{1+\sqrt{1-G^2}}  \\ 
  \frac{-G}{1+\sqrt{1-G^2}}  & 1 
\end{array}} \right)
$$
in the sense that 
$$
\cV \cA \cV^* =  \left( {\begin{array}{*{20}{c}}
   \xi & 0 \\ 
  0 & \xi 
\end{array}} \right)\quad \text{with} \quad \xi:=h \sqrt{1-G^2}=\sqrt{h^2-k^2} >0. 
$$
It is straightforward to verify that $\cV$ always satisfies the compatibility conditions \eqref{eq:Bogoliubov-transformation}. Moreover, $\cV$ is bounded (and hence a Bogoliubov transformation) if and only if $\|G\|=\|kh^{-1}\|<1$  and in this case
\bq  \label{eq:comm-example-norm-V}
\|\cV\| \sim (1- \|G\|)^{-1/4}
\eq
(which means that the ratio between $\|\cV\|$ and $(1- \|G\|)^{-1/4}$ is bounded from above and below by universal positive constants). By Shale's condition \eqref{eq:Shale}, $\cV$ is unitarily implementable if and only if $kh^{-1}$ is Hilbert-Schmidt and in this case, under the conventional form \eqref{eq:u-v-block}, 
\bq  \label{eq:comm-example-HS-V}
\|V\|_{\rm HS} \sim (1- \|G\|)^{-1/4} \|G\|_{\rm HS}.
\eq
Finally, from \eqref{eq:quadratic-Hamiltonian-def4} and the simple estimates 
$$
-\frac{1}{2} k^2 h^{-1} \ge \xi -h = \sqrt{h^2-k^2} -h \ge - k^2 h^{-1}
$$
we deduce that $\bH$ is bounded from below if and only if $k h^{-1/2}$ is Hilbert-Schmidt and in this case
\bq  \label{eq:comm-example-lwb-bH}
\inf \sigma(\bH) \sim - \|k h^{-1/2}\|_{\rm HS}^2.
\eq
  
Thus, in summary, in the above commutative example we have the following {\em optimal} conditions:
\begin{itemize}

\item[$\bullet$] $\cA$ is diagonalized by a Bogoliubov transformation $\cV$ if and only if $\|kh^{-1}\|<1$.  

\item[$\bullet$] $\cV$ is unitarily implementable if and only if  $kh^{-1}$ is Hilbert-Schmidt.

\item[$\bullet$] $\bH$ is bounded from below if and only if $kh^{-1/2}$ is Hilbert-Schmidt.
\end{itemize}
\end{example}

The main message of our work is that
the above necessary and sufficient
conditions in the commutative case
are indeed sufficient also in the
general non-commutative case. Our results are formulated precisely in Theorem \ref{thm:diag} and Theorem \ref{thm:diag-Hamiltonian} in the next section. 

\medskip

\noindent\textbf{Acknowledgements.} We thank Jan Derezi\'nski for several inspiring discussions and useful remarks. We thank the referee for helpful comments. JPS thanks the Erwin Schr\"odinger Institute  for the hospitality during the thematic programme ``Quantum many-body systems, random matrices, and disorder". We gratefully acknowledge the financial supports by the European Union’s Seventh Framework Programme under the ERC Advanced Grant ERC-2012-AdG 321029 (JPS) and the REA grant agreement No. 291734 (PTN), as well as the support of the National Science Center (NCN) grant No. 2012/07/N/ST1/03185 and the Austrian Science Fund (FWF) project Nr. P 27533-N27 (MN).

\section{Main results} \label{sec:main-result}

In this section we state our main results and explain the strategy of the proof. Our first main result concerns diagonalization of block operators.

\begin{theorem}[Diagonalization of bosonic block operators]\label{thm:diag}\text{}\smallskip\\
\noindent {\rm (i) (Existence)}. Let $h:\gH\to \gH$ and $k:\gH\to \gH^*$ be (unbounded) linear operators satisfying $h=h^*>0$, $k^*=J^*kJ^*$ and $D(h)\subset D(k)$. Assume that the operator  $G:=h^{-1/2}J^*kh^{-1/2}$ is densely defined and extends to a bounded operator satisfying $\|G\|<1$. Then we can define the self-adjoint operator
$$
\cA := \left( {\begin{array}{*{20}{c}}
   h & k^* \\ 
  k & J h J^* 
\end{array}} \right)>0 \quad \text{on}~\gH\oplus \gH^*
$$
by Friedrichs' extension. This operator can be diagonalized by a bosonic Bogoliubov transformation $\cV$ on $\gH\oplus \gH^*$ in the sense that
\begin{align*}
\cV \cA \cV^* = \left( {\begin{array}{*{20}{c}}
  \xi & 0  \\ 
  0& J \xi J^* 
\end{array}} \right)
\end{align*}
for a self-adjoint operator $\xi>0$ on $\gH$. Moreover, we have
\bq \label{eq:V-norm-main-thm}
 \|\cV\| \le \left(\frac{1+\|G\|}{1-\|G\|} \right)^{1/4}.
 \eq
\smallskip

\noindent {\rm (ii) (Implementability)}. Assume further that $G$ is Hilbert-Schmidt. Then $\cV$ is unitarily implementable and, under the block form \eqref{eq:u-v-block},
\begin{align} \label{eq:V*V-HS}
\|V\|_{\rm HS} \le \frac{2}{1-\|G\|} \|G\|_{\rm HS}.
\end{align}
\end{theorem}

Next, we consider the diagonalization of quadratic Hamiltonians. 

\begin{theorem} [Diagonalization of quadratic Hamiltonians]\label{thm:diag-Hamiltonian} We keep all assumptions in {\rm Theorem \ref{thm:diag}} (that $\|G\|<1$ and $G$ is Hilbert-Schmidt) and assume further that $kh^{-1/2}$ is Hilbert-Schmidt. Then the quadratic Hamiltonian $\bH$, defined as a quadratic form by \eqref{eq:def-H-gamma-alpha}, is bounded from below and closable, and hence its closure defines a self-adjoint operator which we still denote by $\bH$. Moreover, if $\U_\cV$ is the unitary operator on Fock space implementing the Bogoliubov transformation $\cV$ in Theorem \ref{thm:diag}, then 
\bq \label{eq:diag-H-dGxi}
\bU_\cV \bH \bU_\cV^* = \dGamma(\xi)+ \inf\sigma(\bH).
\eq
Finally, $\bH$ has a unique ground state $\Psi_0=\bU_\cV^*|0\rangle$ whose one-particle density matrices are $\gamma_{\Psi_0}=V^*V$ and $\alpha_{\Psi_0}=JU^*J^*V$ and
\bq \label{eq:inf-H-main-thm}
\inf\sigma(\bH)=\Tr(h^{1/2}\gamma_{\Psi_0} h^{1/2})+\Re \Tr(k^*\alpha_{\Psi_0}) \ge -\frac{1}{2} \| kh^{-1/2}\|_{\rm HS}^2 .
\eq
In particular, $h^{1/2}\gamma_{\Psi_0} h^{1/2}$ and $k^* \alpha_{\Psi_0}$ are trace class. 
\end{theorem}

Before explaining the proof strategy, let us give some remarks.

$\bullet$ Since
$$GG^*\leq \|GG^*\|=\|G\|^2$$
we have
\begin{align} \label{eq:khk*-h} k h^{-1} k^* \le \|G\|^2 JhJ^*.
\end{align}
Thus, the boundedness condition $\|G\|<1$ is a strengthened version of the positivity $\cA>0$. This follows from the following elementary fact whose proof is left to the reader.

\begin{lemma}[Positivity of block operators] \label{lem:positivity-block}\text{}\\
\noindent {\rm (i)} Let $c,d:\gH \to \gH$ and $b:\gH\to \gH^*$ be bounded operators such that $c>0,d\ge 0$ and $b^*=J^*bJ^*$. Then 
$$
D 
: = \left(  {\begin{array}{*{20}c}
   c  & b^*   \\
   b & d  \\
 \end{array} } \right)\ge 0~\text{on}~\gH\oplus \gH^*~\text{if and only if}~b c^{-1} b^* \le d.
$$
\noindent {\rm (ii)} If $c,d,b$ are unbounded operators such that $c,d$ are self-adjoint and $d\ge bc^{-1}b^*$ (which in particular requires that $bc^{-1}b^*$ is densely defined), then $D\ge 0$ as a quadratic form on the domain $D(c^{1/2})\oplus D(d^{1/2})$. Consequently, $D$ can be extended to a non-negative self-adjoint operator by Friedrichs' extension. Moreover, if $b c^{-1} b^*<d$, then $D>0$. 
\end{lemma}  

$\bullet$ The condition $\|G\|<1$ can be interpreted as a non-commutative analogue of the bound $\|kh^{-1}\|<1$ in the commutative case. In general, $\|G\|<1$ is {\em weaker} than $\|kh^{-1}\|<1$ because  
\begin{align*}
(k h^{-1} k^*)^2 &=k h^{-1}k^*kh^{-1} k^* = Jh (h^{-1}k^* J h^{-1} k^* k h^{-1}J^*k h^{-1})hJ^*    \\
&\le \|h^{-1}k^* J h^{-1} k^* k h^{-1}J^*k h^{-1}\| Jh^2J^* \le \| kh^{-1}\|^4 (J h J^*)^2 
\end{align*}
and the square root is operator monotone.

$\bullet$ The implementability condition $\|G\|_{\rm HS}<\infty$  can be interpreted as a non-commutative analogue of the condition $\|kh^{-1}\|_{\rm HS}<\infty$ in the commutative case. In general, these two conditions are not comparable, but if we assume further that $\|kh^{-1/2}\|_{\rm HS}<\infty$ then $\|G\|_{\rm HS}<\infty$ is {\em weaker} than $\|kh^{-1}\|_{\rm HS}<\infty$. Indeed, if $kh^{-1/2}$ is Hilbert-Schmidt then $kh^{-1}k^*$ is trace class and we have the spectral decomposition 
$$J^* kh^{-1}k^*J=\sum_n \lambda_n |u_n\rangle \langle u_n|$$ 
for an eigenbasis $\{u_n\}$ for $\gH$. Therefore, 
\begin{align*}
\|G\|_{\rm HS}^2 & = \Tr(\sqrt{J^* kh^{-1}k^* J} h^{-1}\sqrt{J^* kh^{-1}k^*J}) \\
& = \sum_n \lambda_n \langle u_n, h^{-1} u_n \rangle = \sum_n \langle u_n, J^* kh^{-1}k^* J h^{-1} u_n \rangle \\
& = \Tr(J^* kh^{-1}k^* J h^{-1}) = \Tr ((J^* kh^{-1})^2) \le \| k h^{-1}\|_{\rm HS}^2. 
\end{align*}
Here we have used the identity $\Tr(W^*W)=\Tr(WW^*)$ which holds true for every linear operator $W$. 

$\bullet$ The condition $\|kh^{-1/2}\|_{\rm HS}<\infty$ is the same as in the commutative case. This necessary condition was proved by Bruneau and Derezi\'nski in \cite{BruDer-07} when $k$ is Hilbert-Schmidt. Note that in order to ensure that $\bH$ is bounded from below, we do not really need the conditions $\|G\|<1$ and $\|G\|_{\rm HS}<\infty$ in Theorem \ref{thm:diag}, see Lemma \ref{lem:bounded-below-H}. 

$\bullet$ Our estimates \eqref{eq:V-norm-main-thm}, \eqref{eq:V*V-HS} and \eqref{eq:inf-H-main-thm} are comparable to the formulas \eqref{eq:comm-example-norm-V}, \eqref{eq:comm-example-HS-V} and \eqref{eq:comm-example-lwb-bH} in the commutative case. In particular, the implementability bound \eqref{eq:V*V-HS} provide an upper bound on the particle number expectation of the quasi-free state $\bU_\cV^*|0\rangle$ because $\|V\|_{\rm HS}^2=\langle 0| \bU_{\cV} \mathcal{N} \bU_\cV^*|0\rangle$.

$\bullet$ In Bach and Bru's work \cite{BacBru-15}, they proved \eqref{eq:diag-H-dGxi} under the assumptions 
$$\|kh^{-1}\|<1\quad \text{and}\quad \|kh^{-s}\|_{\rm HS}<\infty~~\text{for all}~~s\in [0,1+\eps].$$
Our Theorem \ref{thm:diag-Hamiltonian} is an improvement of Bach and Bru's result and our conditions seem to be optimal for \eqref{eq:diag-H-dGxi} to hold.

$\bullet$ When $\bH$ is not bounded from below, the diagonalization result on $\cA$ in Theorem \ref{thm:diag}  is of its own interest. In particular, assuming that $\cV$ is unitarily implementable, we can consider the positive self-adjoint operator $\bU_\cV \dGamma(\xi) \bU_\cV^*$ as a renormalized version of $\bH$ (shifted by an infinite constant). It allows us to discuss several important properties of the physical system, such as the (renormalized) ground state $\bU_\cV^*|0\rangle$ and the excitation spectrum, although $\bH$ is not bounded from below. 

\subsection*{Sketch of the proof} The starting point of our approach is to employ a connection between the bosonic diagonalization problem and its fermionic analogue. Such kind of connection has been known for a long time; see Araki \cite{Araki-68} for a heuristic  discussion. To be precise, we will use the following diagonalization result for fermionic block operators.  

\begin{theorem}[Diagonalization of fermionic block operators]\label{thm:diag-fermion}  Let $B$ be a self-adjoint operator on $\gH\oplus \gH^*$ such that $\cJ B \cJ=-B$ and such that $\dim \Ker(B)$ is either even (possibly $0$) or infinite. Then there exists a unitary operator $\cU$ on $\gH \oplus \gH^*$ such that $\cJ\cU \cJ=\cU$ and 
$$
\cU B \cU^* = \left( {\begin{array}{*{20}{c}}
  \xi & 0  \\ 
  0& -J \xi J^* 
\end{array}} \right)
$$
for some operator $\xi \ge0$ on $\gH$. Moreover, if $\Ker(B)=\{0\}$, then $\xi>0$. 
\end{theorem}

By applying Theorem \ref{thm:diag-fermion} to $B=\cA^{1/2}S\cA^{1/2}$, with $S$ given in \eqref{eq:def-S-J}, we can construct the Bogoliubov transformation $\cV$ in Theorem \ref{thm:diag} explicitly:
$$ \cV:= \cU |B|^{1/2} \cA^{-1/2}.$$
This explicit construction is similar to the one used by Simon, Chaturvedi and Srinivasan \cite{SimChaSri-1999} where they offered a simple proof of Williamsons' Theorem. The implementability of $\cV$ is proved using a detailed study of $ \cV^*\cV= \cA^{-1/2}|B|\cA^{-1/2}.$ 

It should be mentioned that when $\bH$ is bounded from below and $\cV$ is unitarily implementable, the desired formula \eqref{eq:diag-H-dGxi},
$$\bU_\cV \bH \bU_\cV = \dGamma(\xi)+ \inf \sigma(\bH),$$
does not follow immediately. In fact, it is not easy to carry out the formal argument in Section \ref{sec:Bogoliubov-transformation} and in particular we could not prove that $\xi-h$ is trace class. Instead, we will use the quadratic form expression of $\bH$ and a key ingredient in our proof is the kinetic energy estimate $\Tr(h^{1/2}V^*Vh^{1/2})<\infty$. 
\\

\noindent\textbf{Organization of the paper.} We will prove Theorems \ref{thm:diag-fermion}, \ref{thm:diag} and \ref{thm:diag-Hamiltonian} in Sections \ref{sec:fermionic}, \ref{sec:diag-A} and \ref{sec:diag-H}, respectively. 

\section{Diagonalization of block operators: fermionic case} \label{sec:fermionic}

In this section we prove Theorem \ref{thm:diag-fermion}. Let us start by recalling a well-known diagonalization result on anti-linear operators. 

\begin{lemma}[Diagonalization of anti-linear operators] \label{lem:e.v.anti-unitary} Let $\gH$ be a separable Hilbert space and let $C$ be an anti-linear operator on $\gH$. Assume that $C=C^*$ and $C^2$ is either compact or equal to $1$. Then $C$ has an orthonormal eigenbasis for $\gH$ with non-negative eigenvalues. 
\end{lemma}

The proof of Lemma \ref{lem:e.v.anti-unitary} is provided in the Appendix. As a consequence, we have

\begin{lemma} \label{lem:even-dim} Let $\mathfrak{H}$ be a separable Hilbert space such that $\dim  \mathfrak{H}$ is even or infinite. Let $\cJ=\cJ^*=\cJ^{-1}$ be an anti-unitary on $\mathfrak{H}$. Then there is a subspace $\mathfrak{K}$ of $\mathfrak{H}$ such that 
$$ \mathfrak{H} = \mathfrak{K}\oplus \cJ \mathfrak{K}.$$
\end{lemma}

\begin{proof} Let $\dim  \mathfrak{H}=2d$ be even or infinite. By Lemma \ref{lem:e.v.anti-unitary}, we can find an orthonormal basis $\{u_n\}_{n=1}^{2d}$ for $\gH$ such that $\cJ u_n=u_n$ for all $n\ge 1$ (note that $1$ is the only non-negative eigenvalue of $\cJ$ because  $\cJ^2=1$). Define
$$ v_{j}:=\frac{u_{2j}+i u_{2j-1}}{\sqrt{2}}.$$
Then the vectors $\{v_j\}_{j=1}^d \cup \{\cJ v_j\}_{j=1}^d$ form an orthonormal basis for  $\mathfrak{H}$. Therefore, $\mathfrak{H}= \mathfrak{K}\oplus \cJ \mathfrak{K}$ with $\mathfrak{K}=\Span\{v_j\}_{j=1}^{d}$.
\end{proof}

Now we are ready to give

\begin{proof}[Proof of Theorem \ref{thm:diag-fermion}] Since $B$ is self-adjoint, we can define the spectral projection $\1(B>0)$ and $\1(B<0)$. By the spectral theorem, we can decompose 
$$
\gH\oplus \gH^* = P_+ \oplus \Ker(B) \oplus P_-
$$
where  
$$P_+:=\1(B>0)(\gH\oplus \gH^*),\quad P_-:=\1(B<0)(\gH\oplus \gH^*).$$
The condition $\cJ B \cJ=-B$ implies that $P_-= \cJ P_+$ and that $\cJ$ leaves $\Ker(B)$ invariant. 
Since $\cJ=\cJ^*=\cJ^{-1}$ is an anti-unitary on $\Ker(B)$ and $\dim \Ker(B)$ is even or infinite, according to Lemma \ref{lem:even-dim}, there is a subspace $\mathfrak{K}$ of $\Ker(B)$ such that 
$$ \Ker(B)= \mathfrak{K}\oplus \cJ \mathfrak{K}.$$
 Thus we have
\bq \label{eq:fermion-spectral-projection}
(P_+\oplus \mathfrak{K}) \oplus \cJ (P_+\oplus \mathfrak{K}) = \gH\oplus \gH^* = (\gH\oplus 0) \oplus  \cJ (\gH\oplus 0).
\eq

Let $W: P_+\oplus \mathfrak{K} \to \gH\oplus 0$ be an arbitrary unitary (which exists because $P_+\oplus \mathfrak{K}$ and $\gH$ have the same dimension). From \eqref{eq:fermion-spectral-projection} it follows that 
$$ \cU := W \oplus \cJ W \cJ$$
is a unitary on $\gH \oplus \gH^*$. It is also clear from the definition of $\cU$ that $\cJ \cU \cJ=\cU$. 

Now we show that $\cU B \cU^*$ is block-diagonal. Note that for every $f\in \gH$, we have $W^*(f\oplus 0)\in P_+ \oplus \mathfrak{K}$, and hence $B W^*(f\oplus 0) \in P_+$ by the  spectral property of $B$. Therefore, 
$$W B W^* (f\oplus 0) \in \gH \oplus 0,\quad \forall f\in \gH.$$
This observation allows us to define a linear operator $\xi:\gH\to \gH$ by
$$
(\xi f) \oplus 0 :=  W B W^* (f\oplus 0),\quad \forall f\in \gH.
$$
Note that $\xi\ge 0$ because
\begin{align} \label{eq:xi>=0}
\langle f, \xi f \rangle = \langle f\oplus 0, (\xi f)\oplus 0 \rangle = \langle W^*(f\oplus 0), B  W^*(f\oplus 0) \rangle \ge 0
\end{align}
for all $f\in \gH$. The last inequality follows from the facts that $W^*(f\oplus 0)\in P_+ \oplus \mathfrak{K}$ and that the restriction of $B$ on $P_+ \oplus \mathfrak{K}$ is nonnegative. 

We will now show that
$$
\cU B \cU^* = \left( {\begin{array}{*{20}{c}}
  \xi & 0  \\ 
  0& -J \xi J^* 
\end{array}} \right)
$$
which is equivalent to
\bq \label{eq:UAU*-fg}
\cU B \cU^* (f\oplus Jg)= (\xi f) \oplus (-J\xi g), \quad \forall f,g\in \gH.
\eq
For every $f\in \gH$, we have $\cU^*(f\oplus 0)=W^* (f\oplus 0) \in P_+ \oplus \mathfrak{K}$. Therefore, $B \cU^*(f\oplus 0) = B W^*(f\oplus 0) \in P_+$ and hence
\bq \label{eq:UAU*f}
\cU B \cU^* (f\oplus 0) = W B W^* (f\oplus 0)= (\xi f)\oplus 0, \quad \forall f\in \gH.
\eq
Similarly, for every $g\in \gH$, we have $\cU^* (0\oplus Jg)= \cJ W^* \cJ (0\oplus Jg) \in P_-\oplus \cJ \mathfrak{K}$. Therefore, $B \cU^* (0\oplus Jg) = B \cJ W^*\cJ (0\oplus Jg) \in P_-$ and hence
\begin{align} \label{eq:UAU*g}
\cU B \cU^* (0\oplus Jg) &= \cJ W \cJ B \cJ W^* \cJ (0\oplus Jg)= - \cJ WBW^* (g\oplus 0) \nn\\
&= -\cJ ((\xi g) \oplus 0) =  - J \xi g, \quad \forall g\in \gH.
\end{align}
Here we have used $\cJ B \cJ=-B$. The desired formula \eqref{eq:UAU*-fg} follows immediately from \eqref{eq:UAU*f} and \eqref{eq:UAU*g}. When $\Ker(B)=\{0\}$, then $\xi>0$ because the inequality \eqref{eq:xi>=0} is strict for every $f\ne 0$.
\end{proof}

\section{Diagonalization of block operators: bosonic case}  \label{sec:diag-A}

\begin{proof}[Proof of Theorem \ref{thm:diag}] (i) {\em Existence}. Let us consider 
$$B:=\cA^{1/2}S\cA^{1/2}.$$
It is clear that $B$ is self-adjoint and $\Ker(B)=\{0\}$ because $\cA>0$. Moreover, $\cJ B \cJ = -B$ because $\cJ\cA \cJ=\cA$ and $\cJ S \cJ =-S$. By applying Theorem \ref{thm:diag-fermion}, we can find a unitary operator $\cU$ on $\gH \oplus \gH^*$ such that $\cJ \cU \cJ=\cU$ and
$$
\cU B \cU^* = \left( {\begin{array}{*{20}{c}}
  \xi & 0  \\ 
  0& -J \xi J^* 
\end{array}} \right)=:D
$$ 
for some self-adjoint operator $\xi>0$ on $\gH$.  

Using \eqref{eq:khk*-h} and Lemma \ref{lem:positivity-block}, we find that 
$$  \delta \cA \le S\cA S \le \delta^{-1}\cA \quad \text{with}\quad \delta=\frac{1-\|G\|}{1+\|G\|}.$$
Therefore,  
\begin{align} \label{eq:BB<AA} \delta \cA^2 \le B^2 =  \cA^{1/2} S \cA S \cA^{1/2} \le \delta^{-1}\cA^2,
\end{align}
and hence 
\begin{align}\label{eq:|B|<A}
\delta^{1/2} \cA \le  |B| \le \delta^{-1/2}\cA
\end{align}
because the square root is operator monotone. Consequently, the operator $|B|^{1/2} \cA^{-1/2}$ is well-defined on $D(\cA)$ and it can be extended to a bounded operator on $\gH \oplus \gH^*$. Thus  
$$ \cV:= \cU |B|^{1/2} \cA^{-1/2}$$
is a bounded operator on $\gH \oplus \gH^*$ and, by \eqref{eq:|B|<A}, 
\bq
\|\cV\| = \| |B|^{-1/2}\cA^{1/2} \| \le  \delta^{-1/4}. \label{eq:normVbound}
\eq
It is straightforward to see that $\cV$ is a Bogoliubov transformation. Indeed, because $\cJ$ commutes with $\cA$, $B$ and $\cU$, it also commutes with $\cV$. Moreover,
\begin{align*}
\cV S \cV^* &= \cU |B|^{1/2} \cA^{-1/2} S \cA^{-1/2} |B|^{1/2} \cU^* \\
&= \cU |B|^{1/2} \cU^* (\cU B^{-1} \cU^*) \cU |B|^{1/2} \cU^* \\
& = |D|^{1/2} D^{-1} |D|^{1/2} = S
\end{align*}
and 
\begin{align*}
\cV^* S \cV &= \cA^{-1/2} |B|^{1/2} \cU^* S \cU |B|^{1/2} \cA^{-1/2}  \\
&= \cA^{-1/2} \cU^* (\cU |B|^{1/2} \cU^*) S (\cU |B|^{1/2} \cU^*) \cU \cA^{-1/2} \\
&=  \cA^{-1/2} \cU^* |D|^{1/2} S |D|^{1/2} \cU \cA^{-1/2} \\
&=  \cA^{-1/2} \cU^* D \cU \cA^{-1/2} = \cA^{-1/2} B \cA^{-1/2} =S .
\end{align*}
Finally, $\cV$ diagonalizes $\cA$ because 
$$ \cV \cA \cV^* = \cU |B| \cU^* = |D| = \left( {\begin{array}{*{20}{c}}
  \xi & 0  \\ 
  0& J\xi J^* 
\end{array}} \right).$$

\noindent (ii) {\em Implementability}. We shall show that Shale's condition $\Tr(V^*V)<\infty$ is equivalent to $\|\cV^*\cV-1\|_{\rm HS}<\infty$ and give explicit upper bounds on $\|\cV^*\cV-1\|_{\rm HS}$ and $\|V\|_{\rm HS}$ in terms of $\|G\|_{\rm HS}$. 
\medskip

\noindent {\bf Step 1.} Our starting point is the formula
\begin{align} \label{eq:V*V-AB} \cV^* \cV =  \cA^{-1/2} |B| \cA^{-1/2}.
\end{align}
From the functional calculus, we have 
\begin{align} \label{eq:functional-calculus-1}
H =\frac{1}{\pi}\int_0^\infty \frac{H^2}{t+H^2} \frac{dt}{\sqrt{t}}
\end{align}
for every self-adjoint operator $H\ge 0$. Therefore, from \eqref{eq:V*V-AB} it follows that 
\begin{align}\label{eq:V*V-integral-formula}
\cV^*\cV-1 &= \cA^{-1/2} (|B|-\cA) \cA^{-1/2} \nn \\
&= \frac{1}{\pi}\int_0^\infty \cA^{-1/2} \left( \frac{B^2}{t+B^2} - \frac{\cA^2}{t+\cA^2} \right) \cA^{-1/2}   \frac{dt}{\sqrt{t}} \nn \\
& = \frac{1}{\pi}\int_0^\infty \cA^{-1/2} \left( \frac{1}{t+\cA^2} - \frac{1}{t+B^2} \right) \cA^{-1/2}   \sqrt{t}dt.
\end{align}
Using the resolvent identity
$$ \frac{1}{t+\cA^2} - \frac{1}{t+B^2} = \frac{1}{t+\cA^2}(B^2-\cA^2) \frac{1}{t+B^2}$$
and the expression
$$
B^2-\cA^2= \cA^{1/2} S \cA S \cA^{1/2}-\cA^2 = \cA^{1/2} E \cA^{1/2}
$$
where
$$ E=E^*:= S\cA S - \cA = -2 \left( {\begin{array}{*{20}{c}}
  0 & k^*  \\ 
  k & 0 
\end{array}} \right),$$
we find that 
\begin{align} \label{eq:resolvent-identity}
\cA^{-1/2} \left( \frac{1}{t+\cA^2} - \frac{1}{t+B^2} \right) \cA^{-1/2} = \frac{1}{t+\cA^2} E \cA^{1/2} \frac{1}{t+B^2} \cA^{-1/2} 
\end{align}
for all $t>0$. The right-side of \eqref{eq:resolvent-identity} can be estimated by the Cauchy-Schwarz inequality 
$$\pm (XY+Y^*X^*) \le \eps^{-1} XX^* + \eps Y^*Y$$
with
$$X:=\frac{1}{t+\cA^2} E \cA^{1/2}|B|^{-1}, \quad Y:=  |B|\frac{1}{t+B^2} \cA^{-1/2}.$$
Note that $XY=Y^*X^*$ because the left-side of \eqref{eq:resolvent-identity} is self-adjoint. Therefore, from \eqref{eq:resolvent-identity} it follows that 
\begin{align} \label{eq:tA-tB}
& \pm 2 \cA^{-1/2} \left( \frac{B^2}{t+B^2} - \frac{\cA^2}{t+\cA^2} \right) \cA^{-1/2} \nn\\
& \le \eps^{-1} \frac{1}{t+\cA^2}E \cA^{1/2} B^{-2} \cA^{1/2} E  \frac{1}{t+\cA^2} + \eps \cA^{-1/2} \frac{B^2}{(t+B^2)^2} \cA^{-1/2} \nn\\
& = \eps^{-1} \frac{1}{t+\cA^2}E S \cA^{-1} S E  \frac{1}{t+\cA^2} + \eps \cA^{-1/2} \frac{B^2}{(t+B^2)^2} \cA^{-1/2}
\end{align}
for all $\eps>0$. Here we have used $B^2=\cA^{1/2}S\cA S\cA^{1/2}$ and $S^{-1}=S$ in the last equality. 
From the functional calculus, we have 
\begin{align}\label{eq:functional-calculus-2}
H= \frac{2}{\pi} \int_0^\infty \frac{H^2}{(t+H^2)^2} \sqrt{t}dt
\end{align}
for every self-adjoint operator $H\ge 0$. Let us integrate \eqref{eq:tA-tB} against the measure $2\pi^{-1}\sqrt{t}dt$ on $(0,\infty)$, then use \eqref{eq:V*V-integral-formula} for the left side and use \eqref{eq:functional-calculus-2} and \eqref{eq:V*V-AB} for the right side. This gives 
\begin{align} \label{eq:K=K1+K2}
\pm 2 (\cV^*\cV-1) \le \eps^{-1} K + \eps \cA^{-1/2}|B| \cA^{-1/2} = \eps^{-1} K + \eps \cV^*\cV
\end{align}
where
\begin{align*}
K &:= \frac{2}{\pi}\int_0^\infty \frac{1}{t+\cA^2}E S \cA^{-1}S  E  \frac{1}{t+\cA^2} \sqrt{t}dt.
\end{align*}

\noindent {\bf Step 2.} Now we show that $K$ is trace class. Using the equality $\Tr(W^*W)=\Tr(WW^*)$, which holds for every operator $W$, and the formula \eqref{eq:functional-calculus-2} we get
\begin{align*} \Tr(K) &= \frac{2}{\pi}\int_0^\infty \Tr \left(\frac{1}{t+\cA^2}E S\cA^{-1}S E  \frac{1}{t+\cA^2} \right) \sqrt{t}dt \\
&= \frac{2}{\pi}\int_0^\infty \Tr \left(\sqrt{E S\cA^{-1} S E} \frac{1}{(t+\cA^2)^2} \sqrt{E S\cA^{-1}S E} \right) \sqrt{t}dt  \\
& = \Tr \left(\sqrt{E S\cA^{-1}S E} \left( \frac{2}{\pi} \int_{0}^\infty \frac{1}{(t+\cA^2)^2} \sqrt{t}dt \right) \sqrt{E S\cA^{-1} S E} \right) \\
& = \Tr \left(\sqrt{E S\cA^{-1} S E} \cA^{-1} \sqrt{E S\cA^{-1}S E} \right).
\end{align*}
By Lemma \ref{lem:positivity-block} and assumption $kh^{-1}k^* \le \|G\|^2 JhJ^*$, we have
$$ \frac{\cA}{1-\|G\|} \ge \cA_0 := \left( {\begin{array}{*{20}{c}}
  h & 0  \\ 
  0 & JhJ^* 
\end{array}} \right)>0,$$
which is equivalent to $ (1-\|G\|)\cA^{-1}\le \cA_0^{-1}$ because the inverse mapping is operator monotone. Thus 
\begin{align} \label{eq:Tr-K-upper}
\Tr(K) & = \Tr \left(\sqrt{E S\cA^{-1} S E} \cA^{-1} \sqrt{E S\cA^{-1}S E} \right) \nn\\
& \le  \frac{1}{1-\|G\|} \Tr \left(\sqrt{E S\cA^{-1}S E} \cA_0^{-1} \sqrt{E S\cA^{-1}S E} \right)\nn \\
& = \frac{1}{1-\|G\|}   \Tr \left(\cA_0^{-1/2} E S\cA^{-1} S E \cA_0^{-1/2} \right) \nn\\
& \le  \frac{1}{(1-\|G\|)^2} \Tr \left(\cA_0^{-1/2} E S \cA_0^{-1} S E \cA_0^{-1/2} \right). 
\end{align}
Here we have used $\Tr(W^*W)=\Tr(WW^*)$ again. From the explicit formulas
$$ E= -2 \left( {\begin{array}{*{20}{c}}
  0 & k^*  \\ 
  k & 0 
\end{array}} \right), \quad \cA_0 = \left( {\begin{array}{*{20}{c}}
  h & 0  \\ 
  0 & JhJ^* 
\end{array}} \right), \quad S = \left( {\begin{array}{*{20}{c}}
  1 & 0  \\ 
  0 & -1 
\end{array}} \right)$$
it is straightforward to compute
\begin{align*}
\cA_0^{-1/2} E S \cA_0^{-1} S E \cA_0^{-1/2} = 4\left( {\begin{array}{*{20}{c}}
  G^*G & 0  \\ 
  0 & J G^*G J^* 
\end{array}} \right).
\end{align*}
Inserting this formula into \eqref{eq:Tr-K-upper} and using  $k^*=J^*k J^*$ we obtain 
\begin{align} \label{eq:Tr-K}
\Tr(K)\le \frac{8}{(1-\|G\|)^2} \Tr (G^*G) <\infty.
\end{align}

\noindent {\bf Step 3.} Now we bound $\|\cV^*\cV-1\|_{\rm HS}$ and $\|V\|_{\rm HS}$ using \eqref{eq:K=K1+K2} and \eqref{eq:Tr-K}. By combining \eqref{eq:K=K1+K2} and the simple bound $\|\cV\| \le \delta^{-1/4}$ (cf. \eqref{eq:normVbound}), we deduce that
\begin{align} \label{eq:cV*cV-K-delta}\pm 2(\cV^*\cV-1) \le \eps^{-1} K +  \eps \delta^{-1/2}
\end{align}
for all $\eps>0$. We have the following general fact whose proof can be found in the Appendix.

\begin{lemma} \label{lem:HS-trace} Let $L=L^*$ be a bounded operator on a separable Hilbert space. If there exists a trace class operator $K \ge 0$ such that
\begin{align} \label{eq:L-K}
\pm 2 L \le \eps^{-1} K + \eps,\quad \forall \eps>0,
\end{align}
then $L$ is Hilbert-Schmidt and $\|L\|_{\rm HS}^2 \le \Tr(K).$
\end{lemma}
From \eqref{eq:Tr-K}, \eqref{eq:cV*cV-K-delta} and Lemma \ref{lem:HS-trace} we conclude that 
\begin{align} \label{eq:cV*cV-1-HS} \|\cV^*\cV-1\|_{\rm HS}^2 \le \delta^{-1/2} \Tr(K) \le \frac{8(1+\|G\|)^{1/2}}{(1-\|G\|)^{5/2}} \|G\|_{\rm HS}^2 .
\end{align}
Note that \eqref{eq:cV*cV-1-HS} implies immediately an upper bound on $\|V\|_{\rm HS}$. Indeed, from the block form \eqref{eq:u-v-block} of $\cV$ and the relations \eqref{eq:relation-U-V} we find that \begin{align} \label{eq:cV*cV-X-Y}
\cV^*\cV- 1 = 2 \left( {\begin{array}{c c}
  X & Y^*  \\ 
  Y & J X J^* 
\end{array}} \right)
\end{align}
where 
$$X:=V^*V \quad \text{and}\quad Y:=J V^*JU = JU^*J^*V.$$  
Using \eqref{eq:relation-U-V} again we have 
\bq \label{eq:CC-VV-VVVV}
Y^*Y = (JU^*J^*V)^* JU^*J^*V = V^* J UU^* J^* V = X + X^2 .
\eq
Consequently,
$$ \|\cV^*\cV-1\|_{\rm HS}^2 \ge 4 \|Y\|_{\rm HS}^2 \ge 4\Tr(X) = 4 \|V\|^2_{\rm HS} .$$
Therefore, from \eqref{eq:cV*cV-1-HS} we can conclude that $V$ is Hilbert-Schmidt and 
\begin{align*} \|V\|_{\rm HS}^2 \le \frac{2(1+\|G\|)^{1/2}}{(1-\|G\|)^{5/2}} \|G\|_{\rm HS}^2.
\end{align*}
To obtain the better bound \eqref{eq:V*V-HS} we need to analyze \eqref{eq:K=K1+K2} more carefully. We will need the following technical result.

\begin{lemma} \label{lem:diag-X-Y} Assume that $X:\gH\to \gH$ is a nonnegative trace class operator and $Y:\gH \to \gH^*$ is a Hilbert-Schmidt operator satisfying 
$$Y^*=J^* YJ^*, \quad Y^*Y=X+X^2.$$
Then there exists an orthonormal basis $\{u_n\}_{n\ge 1}$ for $\gH$ and non-negative numbers $\{\lambda_n\}_{n\ge 1}$ such that
\bq \label{eq:ev-V*V-C}X u_n =\lambda_n u_n, \quad Y u_n = J Y^* J u_n = \sqrt{\lambda_n+\lambda_n^2} J u_n, \quad \forall n\ge 1.
\eq
\end{lemma}

\begin{proof}[Proof of Lemma \ref{lem:diag-X-Y}] Note that $C:=J^*Y$ is an anti-linear operator on $\gH$ satisfying $C=C^*$ and $C^2=Y^*Y=X+X^2$. Thus both $X$ and $C^2$ are trace class and they have common eigenspaces. For every such eigenspace $W$, since $C$ leaves $W$ invariant, it follows from Lemma \ref{lem:e.v.anti-unitary} that $C$ admits an orthonormal eigenbasis for $W$ with non-negative eigenvalues. Thus there exists an orthonormal basis $\{u_n\}_{n\ge 1}$ for $\gH$ and non-negative numbers $\{\lambda_n\}_{n\ge 1}$ such that
$X u_n =\lambda_n u_n$ and $C u_n = \sqrt{\lambda_n+\lambda_n^2} u_n$ for all $n$. Here we have used $C^2=X+X^2$ for the relation between eigenvalues of $X$ and eigenvalues of $C$. The conclusion \eqref{eq:ev-V*V-C} then follows immediately. 
\end{proof}

Now let $\{u_n\}_{n\ge 1}$ satisfy \eqref{eq:ev-V*V-C}. Using \eqref{eq:cV*cV-X-Y} we can rewrite \eqref{eq:K=K1+K2} as
$$
2(2-\eps) \left( {\begin{array}{c c}
  X & Y^*  \\ 
  Y & J X J^* 
\end{array}} \right) \le \eps^{-1}K+ \eps, \quad \forall \eps>0.
$$
Taking the expectation against $u_n\oplus Ju_n$ we find that
\begin{align} \label{eq:lambdan-eps}
4(2-\eps) (\lambda_n + \sqrt{\lambda_n+\lambda_n^2}) \le \eps^{-1} \langle u_n \oplus Ju_n, K u_n\oplus Ju_n \rangle + 2\eps
\end{align}
for all $n \ge 1$ and for all $\eps>0$. We will show that
\begin{align} \label{eq:lambdan-Kn}
\lambda_n \le \frac{1}{4} \langle u_n \oplus Ju_n, K u_n\oplus Ju_n \rangle, \quad \forall n\ge 1.
\end{align}
Indeed, if $\lambda_n\le 1/4$, then \eqref{eq:lambdan-eps} implies that 
$$
8\sqrt{\lambda_n} -2\eps \le 4(2-\eps)\sqrt{\lambda_n} \le \eps^{-1}\langle u_n \oplus Ju_n, K u_n\oplus Ju_n \rangle + 2\eps,
$$
and hence
$$8\sqrt{\lambda_n} \le \eps^{-1}\langle u_n \oplus Ju_n, K u_n\oplus Ju_n \rangle + 4\eps,\quad \forall \eps>0.$$
Optimizing over $\eps>0$ leads to \eqref{eq:lambdan-Kn}. On the other hand, if $\lambda_n\ge 1/4$, then we can choose $\eps=1/2$ and deduce from \eqref{eq:lambdan-eps}  that 
$$
1+ 8 \lambda_n \le 4(2-\eps)(\lambda_n+ \sqrt{\lambda_n+\lambda_n^2}) \le  2\langle u_n \oplus Ju_n, K u_n\oplus Ju_n \rangle + 1
$$
which is equivalent to \eqref{eq:lambdan-Kn}. Summing \eqref{eq:lambdan-Kn} over $n\ge 1$ and using the fact that $\{2^{-1/2}u_n \oplus Ju_n\}_{n\ge 1}$ is an orthonormal family in $\gH\oplus \gH^*$  we get
$$
\Tr(V^*V)= \sum_n \lambda_n \le \frac{1}{4}\sum_n \langle u_n \oplus Ju_n, K u_n\oplus Ju_n \rangle \le \frac{1}{2} \Tr(K). $$
Inserting the upper bound \eqref{eq:Tr-K} into the latter estimate we obtain \eqref{eq:V*V-HS}. 
\end{proof}

\section{Diagonalization of quadratic Hamiltonians} \label{sec:diag-H}

In this section we will prove Theorem \ref{thm:diag-Hamiltonian}. Recall that $\bH$ is defined as a quadratic form by \eqref{eq:def-H-gamma-alpha}
$$
\langle \Psi, \bH \Psi \rangle = \Tr(h^{1/2} \gamma_\Psi h^{1/2})+ \Re \Tr(k^* \alpha_\Psi)
$$ 
where $\gamma_\Psi$ and $\alpha_\Psi$ are the one-particle density matrices of $\Psi$ defined by \eqref{eq:def-gamma-alpha}. From the CCR \eqref{eq:CCR}, we see that $\alpha^*=J^* \alpha J^*$ and
\begin{align} \label{eq:def-Gamma-AF*AF}
\left\langle F, \left( {\begin{array}{c c}
  \gamma_\Psi & \alpha_\Psi^*  \\ 
  \alpha_\Psi & 1+ J \gamma_\Psi J^* 
\end{array}} \right)F \right \rangle = \langle \Psi, A^*(F) A(F) \Psi \rangle, \quad \forall F\in \gH \oplus \gH^*,
\end{align}
where the generalized creation and annihilation operators $A^*(F)$, $A(F)$ are defined by \eqref{eq:def-A=a+a*}. Thus 
\begin{align} \label{eq:Gamma>=0}
\left( {\begin{array}{c c}
  \gamma_\Psi & \alpha_\Psi^*  \\ 
  \alpha_\Psi & 1+ J \gamma_\Psi J^* 
\end{array}} \right) \ge 0.
\end{align}
By Lemma \ref{lem:positivity-block}, it follows that 
\bq \label{eq:relation-gamma-alpha}
\gamma_\Psi \ge 0\quad \text{and} \quad \gamma_\Psi \ge  \alpha_\Psi^* (1+J\gamma_\Psi J^*)^{-1} \alpha_\Psi. 
\eq
Consequently, since $$(1+J\gamma_\Psi J^*)^{-1}\geq (1+\|\gamma_\Psi\|)^{-1}$$ we have    
\bq
\gamma_\Psi(1+\|\gamma_\Psi\|)\geq \alpha_\Psi^* \alpha_\Psi. \nn
\eq 
Furthermore, if $\Psi$ has finite particle number expectation, namely $\langle \Psi, \mathcal{N} \Psi \rangle<\infty$, then $\gamma_\Psi$ is trace class and $\alpha_\Psi$ is Hilbert-Schmidt. 

Note that from \eqref{eq:relation-gamma-alpha} and the assumption $\|k h^{-1/2}\|_{\rm{HS}}<\infty$, it follows that if both $\gamma_\Psi$ and $h^{1/2}\gamma_\Psi h^{1/2}$ are trace class, then $k \alpha_\Psi^*=k h^{-1/2}\cdot h^{1/2} \alpha_\Psi^* $ is trace class (because $k h^{-1/2}$ and $h^{1/2} \alpha_\Psi^*$ are Hilbert-Schmidt), and hence the right side of \eqref{eq:def-H-gamma-alpha} is well-defined and finite. 



Now we start proving Theorem \ref{thm:diag-Hamiltonian}. The first step is to verify that $\bH$ is bounded from below.

\begin{lemma}[Lower bound on $\bH$] \label{lem:bounded-below-H} Assume that $h>0$, $k^*=J^*kJ^*$,  $kh^{-1}k^* \le JhJ^*$ and $\Tr(kh^{-1}k^*)<\infty$. Then  
$$\bH \ge -\frac{1}{2} \Tr(kh^{-1}k^*)$$
as a quadratic form. Moreover, if $kh^{-1}k^* \le \delta JhJ^*$ for some $0\le \delta<1$, then 
$$ (1+\sqrt{\delta}) \dGamma(h) + \frac{\sqrt{\delta}}{2} \Tr( k h^{-1} k^*) \ge \bH \ge (1+\sqrt{\delta}) \dGamma(h) - \frac{\sqrt{\delta}}{2} \Tr( k h^{-1} k^*)  $$
as quadratic forms. Consequently, the quadratic form $\bH$ defines a self-adjoint operator,  still denoted by $\bH$, such that $\inf \sigma(\bH) \ge -\frac{1}{2} \Tr(kh^{-1}k^*)$.
\end{lemma} 

Note that we do not require that $G=h^{-1/2}J^*kh^{-1/2}$ is Hilbert-Schmidt in Lemma \ref{lem:bounded-below-H}. The first lower bound in Lemma \ref{lem:bounded-below-H} was proved in \cite[Theorem 5.4]{BruDer-07} under the additional condition that $k$ is Hilbert-Schmidt. Our proof below is different from  \cite{BruDer-07} and it does not require the boundedness of $k$.

\begin{proof}[Proof of Lemma \ref{lem:bounded-below-H}] Let $\Psi$ be a normalized vector in the domain $\mathcal{Q}$ defined in \eqref{eq:def-mcQ}. Since $\gamma_\Psi$ and $\alpha_\Psi$ are finite-rank operators, we can use the cyclicity of the trace and the Cauchy-Schwarz inequality to write
\begin{align} 
|\Tr( k^* \alpha_\Psi )| &= |\Tr( \alpha_\Psi k^*)| = |\Tr((1+J\gamma_\Psi J^*)^{-1/2}\alpha_\Psi h^{1/2}h^{-1/2} k^* (1+J\gamma_\Psi J^*)^{1/2})| \nn \\
&\le \|(1+J\gamma_\Psi J^*)^{-1/2}\alpha_\Psi h^{1/2}\|_{\rm HS} \| h^{-1/2} k^* (1+J\gamma_\Psi J^*)^{1/2} \|_{\rm HS} \nn \\
&= \left[ \Tr( h^{1/2}  \alpha_\Psi^* (1+J\gamma_\Psi J^*)^{-1}\alpha_\Psi  h^{1/2} ) \right]^{1/2} \nn \\
&\quad \times \left[ \Tr \Big( (1+J\gamma_\Psi J^*)^{1/2} k h^{-1} k^* (1+J\gamma_\Psi J^*)^{1/2} \Big) \right]^{1/2} . \label{eq:|Tr-k*-alpha-Psi|}
\end{align}
Using $\alpha_\Psi^* (1+J\gamma_\Psi J^*)^{-1}\alpha_\Psi\le \gamma_\Psi$ and $k h^{-1} k^* \le JhJ^*$, we get
\begin{align}
|\Tr(k^* \alpha_\Psi)| &\le \left[\Tr(h^{1/2} \gamma_\Psi h^{1/2}) \right]^{1/2}  \left[ \Tr( k h^{-1} k^*)+ \Tr(h^{1/2} \gamma_\Psi h^{1/2}) \right]^{1/2} \nn\\
& \le \Tr(h^{1/2} \gamma_\Psi h^{1/2}) + \frac{1}{2} \Tr( k h^{-1} k^*). \label{eq:|Tr-k*-alpha-Psi|-b}
\end{align}
Here in the last estimate we have used $\sqrt{x(x+y)}\le x+y/2$ for real numbers $x,y\ge 0$. 
Thus by definition \eqref{eq:def-H-gamma-alpha}, we get
$$
\langle \Psi, \bH \Psi \rangle = \Tr(h^{1/2} \gamma_\Psi h^{1/2}) + \Re \Tr(k^* \alpha_\Psi) \ge -\frac{1}{2} \Tr( k h^{-1} k^*)
$$
for $\Psi\in \mathcal{Q}$. Thus $\bH \ge -(1/2)\Tr( k h^{-1} k^*)$.  

Now we assume further that $kh^{-1}k^* \le \delta JhJ^*$ for some $0\le \delta<1$. Then inserting this bound into \eqref{eq:|Tr-k*-alpha-Psi|}, we can improve \eqref{eq:|Tr-k*-alpha-Psi|-b} to 
\begin{align*}
|\Tr(k^* \alpha_\Psi)| \le \sqrt{\delta} \Big( \Tr(h^{1/2} \gamma_\Psi h^{1/2}) + \frac{1}{2} \Tr( k h^{-1} k^*) \Big).
\end{align*}
Therefore, 
\begin{align*}
\langle \Psi, \bH \Psi \rangle &= \Tr(h^{1/2} \gamma_\Psi h^{1/2}) + \Re \Tr(k^* \alpha_\Psi) \\
& \ge (1-\sqrt{\delta}) \Tr(h^{1/2} \gamma_\Psi h^{1/2})  -\frac{\sqrt{\delta}}{2} \Tr( k h^{-1} k^*) \\
&= (1-\sqrt{\delta}) \langle \Psi, \dGamma(h) \Psi\rangle  -\frac{\sqrt{\delta}}{2} \Tr( k h^{-1} k^*),
\end{align*}
and similarly,
\begin{align*}
\langle \Psi, \bH \Psi \rangle  \le (1+\sqrt{\delta}) \langle \Psi, \dGamma(h) \Psi\rangle  + \frac{\sqrt{\delta}}{2} \Tr( k h^{-1} k^*)
\end{align*}
for $\Psi\in \mathcal{Q}$. In summary, 
$$ (1+\sqrt{\delta}) \dGamma(h) + \frac{\sqrt{\delta}}{2} \Tr( k h^{-1} k^*) \ge \bH \ge (1+\sqrt{\delta}) \dGamma(h) - \frac{\sqrt{\delta}}{2} \Tr( k h^{-1} k^*)  $$
as quadratic forms. Hence, the form domain of $\bH$ can be extended to be the same with that of $\dGamma(h)$, which is closed because $\dGamma(h)$ is a self-adjoint operator. Since the quadratic form $\bH$ is bounded from below and closable, its closure defines a self-adjoint operator by \cite[Theorem VIII.15]{ReeSim-80}. The quadratic form bound $\bH \ge -(1/2)\Tr( k h^{-1} k^*)$ remains valid as an operator bound. 
\end{proof}

In order to prove the diagonalization formula \eqref{eq:diag-H-dGxi} and the identity \eqref{eq:inf-H-main-thm}, we need to show that both $h^{1/2}V^*Vh^{1/2}$ and $k^*JU^*J^*V$ are trace class. To this end, we will use the following lemma which is inspired by Grillakis and Machedon \cite[eq. (42)]{GriMac-12} (see also \cite{NamNap-15}). Recall that $\cA$ and $\cV$ have the block forms \eqref{eq:u-v-block} and \ref{eq:def-A-block}, respectively. 

\begin{lemma}[Diagonalization equations] \label{lem:diag-eq} If $\cV \cA \cV^*$ is block diagonal, then 
\begin{equation} \label{eq:linear}
\left\{
\begin{aligned}
h X-Xh+k^*Y-Y^* k &=0, \\
J h J^* Y + Y h + k X + J X J^* k + k &=0. 
\end{aligned}
\right.
\end{equation}
where $X = V^*V$ and $Y = JV^*JU = JU^*J^*V$. 
\end{lemma}

Note that by computing the off-diagonal term of $\cV \cA \cV^*$ we get 
\begin{align} \label{eq:nonlinear}
Uh V^* + J^* V h U^* J^* + U k^* J U^* J^* + J^* V J^* k V^* =0
\end{align}
which is essentially equivalent to \eqref{eq:linear}. However, the equations \eqref{eq:linear} are easier to analyze because they are {\em linear} (in terms of $X$ and $Y$), while \eqref{eq:nonlinear} is {\em nonlinear} (in terms of $U$ and $V$). 

\begin{proof}[Proof of Lemma \ref{lem:diag-eq}] Recall that  $\cV^{-1}=S\cV^* S$. Since $\cV \cA \cV^*=\cV \cA S \cV^{-1}S$ is block diagonal, $\cV \cA S \cV^{-1}$ is also block diagonal. Therefore, $[\cV \cA S \cV^{-1},S]=0$, and hence $[\cA S , \cV^{-1}S\cV ]=0$. Using $\cV^{-1}=S\cV^* S$ again and \eqref{eq:cV*cV-X-Y} we have
\bq
\cV^{-1}S\cV  = S\cV^{*}\cV 
=  \left( {\begin{array}{c c}
  1+2X & 2Y^*  \\ 
  - 2Y & -(1+ 2J X J^*) 
\end{array}} \right).
\eq
Therefore,
\begin{align*}
& 0 = [\cA S, \cV^{-1}S\cV] =  \left( {\begin{array}{c c}
  h & -k^*   \\ 
  k  & - J h J^* 
\end{array}} \right) \left( {\begin{array}{c c}
  1+2X & 2Y^*  \\ 
  - 2Y & -(1+ 2J X J^*) 
\end{array}} \right) \\
& \quad \quad\quad\quad\quad\quad\quad\quad - \left( {\begin{array}{c c}
  1+2X & 2Y^*  \\ 
  - 2Y & -(1+ 2J X J^*) 
\end{array}} \right) \left( {\begin{array}{c c}
  h & -k^*   \\ 
  k  & - J h J^* 
\end{array}} \right) \\
&= 2\left( {\begin{array}{*{20}{c}}
  h X-Xh+k ^*Y-Y^* k & h Y^*+ Y^* J h J^* +  k^* JXJ^* + X k^* + k^* \\ 
  J h J^* Y + Y h + k X + J X J^* k + k &  J h X J^* - J X h J^* + k Y^* - Y k ^* 
\end{array}} \right) 
\end{align*}
and \eqref{eq:linear} follows immediately. 
\end{proof}

Now we are ready to prove
\begin{lemma} \label{lem:Tr-hX} Under conditions of Theorem \ref{thm:diag-Hamiltonian}, the operators $h^{1/2}X h^{1/2}$ and $k^*Y$ are trace class, where $X = V^*V$ and $Y = JV^*JU = JU^*J^*V$. 
\end{lemma}

\begin{proof} By Lemma \ref{lem:diag-X-Y} we can find an orthonormal basis $\{u_n\}_{n\ge 1}$ for $\gH$ and non-negative numbers $\{\lambda_n\}_{n\ge 1}$ such that
$$ X u_n =\lambda_n u_n, \quad Y u_n = J Y^* J u_n = \sqrt{\lambda_n+\lambda_n^2} J u_n, \quad \forall n\ge 1.
$$
Using the second equation in \eqref{eq:linear}, we obtain
\begin{align*} 
0&= \langle Ju_n, (J h J^* Y + Y h + k X + J X J^* k + k) u_n \rangle \nn\\
&= 4\sqrt{\lambda_n+\lambda_n^2} \langle u_n, h u_n \rangle + 2(1+2 \lambda_n) \langle Ju_n, ku_n\rangle, \quad \forall n\ge 1.
\end{align*}
By the Cauchy-Schwarz inequality
$$ |\langle Ju_n, ku_n\rangle|= |\langle h^{-1/2}k^*Ju_n, h^{1/2}u_n \rangle| \le \| h^{-1/2}k^*Ju_n \|\cdot \| h^{1/2}u_n\|$$ 
we find that 
$$ \lambda_n \langle u_n, h u_n \rangle \le \frac{(1+2\lambda_n)^2}{4+4\lambda_n} \| h^{-1/2}k^*Ju_n \|^2. $$
for all $n\ge 1$. Summing the latter estimate over $n\ge 1$ and using
$$ \frac{(1+2\lambda_n)^2}{4+4\lambda_n} \le \lambda_n +\frac{1}{4} \le \|X\|+\frac{1}{4}$$
we obtain
\begin{align*} \Tr(h^{1/2}X h^{1/2}) & =\Tr(X^{1/2}hX^{1/2})  =\sum_{n} \lambda_n \langle u_n, h u_n \rangle \\
&\le \left( \|X\|+\frac{1}{4} \right) \sum_n  \| h^{-1/2}k^*Ju_n \|^2 \\ 
&=  \left( \|X\|+\frac{1}{4} \right) \Tr(kh^{-1}k^*) <\infty. 
\end{align*}
Thus $h^{1/2}X h^{1/2}$ is a trace class operator. Moreover, by \eqref{eq:CC-VV-VVVV}, $Y^*Y=X+X^2 \le (1+\|X\|) X$, and hence $h^{1/2}Y^*Yh^{1/2}$ is also a trace class operator. Consequently, $k^*Y= J^*k h^{-1/2} \cdot h^{1/2} Y^* J$ is a trace class operator because $k h^{-1/2}$ and $h^{1/2} Y^*$ are Hilbert-Schmidt operators. 
\end{proof}

Now we are ready to finish

\begin{proof}[Proof of Theorem \ref{thm:diag-Hamiltonian}] From Lemma \ref{lem:bounded-below-H}, we know that the quadratic form $\bH$ defines a self-adjoint operator, still denoted by $\bH$, such that 
$$\inf \sigma(\bH) \ge -\frac{1}{2} \Tr(kh^{-1}k^*).$$
Let $\cV$ be as in Theorem \ref{thm:diag}. Let $\Psi$ be a normalized vector in $\mathcal{Q}$ defined in \eqref{eq:def-mcQ}. Consider $\Psi':=\bU_\cV^* \Psi$. From \eqref{eq:def-Gamma-AF*AF} and \eqref{eq:V-Uv-action}, it is straightforward to see that  
\begin{align}\label{eq:V*-Gamma-V} \left( {\begin{array}{*{20}{c}}
  \gamma_{\Psi'} & \alpha_{\Psi'}^*  \\ 
  \alpha_{\Psi'} & 1+ J \gamma_{\Psi'} J^* 
\end{array}} \right) = \cV^* \left( {\begin{array}{*{20}{c}}
  \gamma_{\Psi} & \alpha_{\Psi}^*  \\ 
  \alpha_{\Psi} & 1+ J \gamma_{\Psi} J^* 
\end{array}} \right) \cV.
\end{align}
Moreover, we have
\begin{align} \label{eq:V*-0001-V}
\cV^* 
\left( {\begin{array}{c c}
  0 & 0  \\ 
  0 & 1 
\end{array}} \right) \cV = \left( {\begin{array}{c c}
  X & Y^*  \\ 
  Y & 1+JXJ^* 
\end{array}} \right)
\end{align}
with $X = V^*V$ and $Y = JV^*JU = JU^*J^*V$. From \eqref{eq:V*-Gamma-V} and \eqref{eq:V*-0001-V}, we obtain 
\begin{align*}
\left( {\begin{array}{c c}
  \gamma_{\Psi'} - X & \alpha_{\Psi'}^* - Y^*  \\ 
  \alpha_{\Psi'}-Y & J (\gamma_{\Psi'}-X) J^* 
\end{array}} \right)  
&= \left( {\begin{array}{c c}
  \gamma_{\Psi'} & \alpha_{\Psi'}^*  \\ 
  \alpha_{\Psi'} & 1+ J \gamma_{\Psi'} J^* 
\end{array}} \right) -   \left( {\begin{array}{c c}
 X  & Y^*   \\ 
  Y & 1+ J X J^* \end{array}} \right) \\
& = \cV^* \left( {\begin{array}{c c}
  \gamma_{\Psi} & \alpha_{\Psi}^*  \\ 
  \alpha_{\Psi} & 1+ J \gamma_{\Psi} J^* 
\end{array}} \right) \cV - \cV^* 
\left( {\begin{array}{c c}
  0 & 0  \\ 
  0 & 1 
\end{array}} \right) \cV \\
&= \cV ^* \left( {\begin{array}{c c}
  \gamma_{\Psi} & \alpha_{\Psi}^*  \\ 
  \alpha_{\Psi} & J \gamma_{\Psi} J^* 
\end{array}} \right) \cV .
\end{align*}

Recall that $\gamma_\Psi$ and $\alpha_\Psi$ are finite-rank operators because $\Psi\in \mathcal{Q}$. Therefore,$\gamma_{\Psi'} - X$ and  $\alpha_{\Psi'}-Y$ are also finite-rank operators. Using the cyclicity of the trace we find that
\begin{align*} 
& \Tr(h^{1/2}(\gamma_{\Psi'}- X) h^{1/2}) + \Re \Tr(k^* (\alpha_{\Psi'}-Y)) \\
& = \frac{1}{2}  \Tr \left[ \cA  \left( {\begin{array}{c c}
  \gamma_{\Psi'} - X & \alpha_{\Psi'}^* - Y^*  \\ 
  \alpha_{\Psi'}-Y & J (\gamma_{\Psi'}-X) J^* 
\end{array}} \right)  \right] \\
& = \frac{1}{2}\Tr \left[ \cA  \cV^* \left( {\begin{array}{c c}
  \gamma_{\Psi} & \alpha_{\Psi}^* \\ 
  \alpha_{\Psi} & J \gamma_{\Psi}J^* 
\end{array}} \right) \cV  \right] 
 = \frac{1}{2}\Tr \left[ \cV \cA  \cV^* \left( {\begin{array}{c c}
  \gamma_{\Psi} & \alpha_{\Psi}^* \\ 
  \alpha_{\Psi} & J \gamma_{\Psi}J^* 
\end{array}} \right)  \right] \\
& = \frac{1}{2} \Tr \left[  \left( {\begin{array}{c c}
  \xi & 0 \\ 
  0 & J \xi  J^* 
\end{array}} \right) \left( {\begin{array}{c c}
  \gamma_{\Psi} & \alpha_{\Psi}^* \\ 
  \alpha_{\Psi} & J \gamma_{\Psi}J^* 
\end{array}} \right)  \right] = \Tr(\xi \gamma_\Psi) = \langle \Psi, \dGamma(\xi) \Psi \rangle.
\end{align*}
Thus by the quadratic form expression \eqref{eq:def-H-gamma-alpha}, we have
\begin{align*}
\langle \Psi, \bU_\cV \bH \bU_\cV^* \Psi \rangle &= \langle \bU_\cV^* \Psi, \bH \bU_\cV^* \Psi \rangle = \Tr(h^{1/2}\gamma_{\Psi'}h^{1/2})+ \Re \Tr(k^* \alpha_{\Psi'}) \\ 
&= \Tr(h^{1/2}(\gamma_{\Psi'}- X) h^{1/2}) + \Re \Tr(k^* (\alpha_{\Psi'}-Y)) \\
&\quad + \Tr(h^{1/2}Xh^{1/2})+ \Re \Tr(k^* Y) \\
&= \langle \Psi, \dGamma(\xi) \Psi \rangle + \Tr(h^{1/2}Xh^{1/2})+ \Re \Tr(k^* Y)
\end{align*}
for all $\Psi \in \mathcal{Q}$. Recall that we have proved in Lemma \ref{lem:Tr-hX} that $h^{1/2}Xh^{1/2}$ and $k^* Y$ are trace class operators. Hence,  
$$ \bU_\cV \bH \bU_\cV^* = \dGamma(\xi) + \Tr(h^{1/2}X h^{1/2})+ \Re \Tr(k^* Y). $$
Since $\dGamma(\xi)$ has a unique ground state $|0\rangle$ with the ground state energy $0$, we conclude that $\bH$ has a unique ground state $\Psi_0=\bU_\cV^* |0\rangle$ with the ground state energy 
$$ \inf\sigma(\bH) = \Tr(h^{1/2}Xh^{1/2})+ \Re \Tr(k^*Y).$$
Finally, using \eqref{eq:V*-Gamma-V} and \eqref{eq:V*-0001-V} we find that $\gamma_{\Psi_0}=X$ and $\alpha_{\Psi_0}=Y$. 
\end{proof}

\section{Appendix}

In this appendix we prove two abstract results in functional analysis. 

\begin{proof}[Proof of Lemma \ref{lem:e.v.anti-unitary}] Let us first consider the case when $\mathfrak{H}$ is finite dimensional. Since $C^2=C^*C$ is non-negative, it has an eigenvalue $\mu \ge 0$ with an eigenvector $u\ne 0$. Since $(C+\sqrt{\mu})(C-\sqrt{\mu})u=0$, we have either $Cu = \sqrt{\mu}u$ or $Cv = \sqrt{\mu}v$ with $v=i(C-\sqrt{\mu})u \ne 0$. Thus $C$ has an eigenvector $u_1\in \gH$ with a non-negative eigenvalue. Since $C=C^*$, it leaves the orthogonal subspace $\{u_1\}^\bot$ invariant. By the previous argument, $C$ has an eigenvector $u_2\in \{u_1\}^\bot$ with a non-negative eigenvalue. By iterating this process, we see that $C$ has an orthonormal eigenbasis for $\gH$ with non-negative eigenvalues.

Now we consider the case when $\mathfrak{H}$ is infinite dimensional. Assume that $C^2$ is compact. Then by the spectral theorem $C^2$ has an orthonormal eigenbasis for $\gH$. Note that $C$ leaves every eigenspace of $C^2$ invariant. Therefore, we can apply the result on the finite dimensional case to every eigenspace of $C^2$ except $\Ker(C^2)$. Moreover, $\Ker(C^2)=\Ker(C^*C)=\Ker(C)$. The desired conclusion follows immediately. 

Now assume that $C^2=1$. Take an arbitrary trace class operator $K>0$ on $\mathfrak{H}$ and consider the operator $K_1:=K+C^* K C$. Note that $K_1$ is also trace class and $K_1>0$. By the spectral theorem, $K_1$ admits an orthonormal eigenbasis for $\mathfrak{H}$ and every eigenspace of $K_1$ is finite dimensional. On the other hand, $C K_1 = K_1 C$ because $C=C^*$ and  $C^2=1$. Therefore, $C$ leaves every eigenspace of $K_1$ invariant. Then applying the result on the finite dimensional case for every eigenspace of $K_1$ we get the desired  conclusion.            
\end{proof}

\begin{proof}[Proof of Lemma \ref{lem:HS-trace} ] By taking the expectation of \eqref{eq:L-K} against a vector $u$ and optimizing over $\eps>0$, we obtain 
\begin{align} \label{eq:K-K1}
|\langle u, L u \rangle| \le  \langle u, K u \rangle^{1/2} \|u\|, \quad \forall u.
\end{align}

Let us show that $L$ is a compact operator. It suffices to show that $Lu_n\to 0$ strongly for every sequence $u_n\wto 0$ weakly. We first consider the case $L \ge 0$. Since $u_n\wto 0$ weakly, $\|u_n\|$ is bounded by the principle of uniform boundedness. Moreover, $\sqrt{K}u_n\to 0$ strongly because $\sqrt{K}$ is compact. Therefore, from \eqref{eq:K-K1} we get
$$
\|Lu_n\|^2 \le \|\sqrt{L}\|^2 \cdot \|\sqrt{L}u_n\|^2 \le \|\sqrt{L} \|^2 \cdot \|\sqrt{K}u_n\|\cdot \|u_n\| \to 0.
$$
In the general case, when $L$ is not necessarily positive, from \eqref{eq:L-K} we have
\begin{align*}
2L \1(L\ge 0) \le \eps^{-1} \1(L\ge 0) K \1(L\ge 0) + \eps,\quad \forall \eps>0.  
\end{align*}
Since $\1(L\ge 0) K \1(L\ge 0)$ is trace class and $L \1(L\ge 0) \ge 0$, we conclude that $L \1(L\ge 0)$ is compact. Similarly, $L \1(L<0)$ is compact, and hence $L=L\1(L\ge 0) + L \1(L< 0)$ is compact.  

Finally, since $L=L^*$ is a compact operator, it has an orthonormal eigenbasis $\{v_n\}$ with real eigenvalues. Using \eqref{eq:K-K1}  again we obtain
$$ \Tr(L^2)= \sum_n |\langle v_n, L v_n \rangle|^2 \le \sum_n  \langle v_n, K v_n \rangle = \Tr(K).$$
\end{proof}

\bibliographystyle{siam}

\end{document}